\documentclass[11pt]{article}
\usepackage{fullpage}
\usepackage{authblk}
\usepackage{amsmath}
\usepackage{amsthm}
\usepackage{amssymb}
\usepackage{amsfonts}
\usepackage{mathtools}
\usepackage{graphicx}
\usepackage{xcolor}
\usepackage{algorithm}
\usepackage[noend]{algpseudocode}
\theoremstyle{plain}
\newtheorem{theorem}{Theorem}

\newtheorem{claim}{Claim}

\newtheorem{proposition}{Proposition}
\newtheorem{corollary}{Corollary}

\theoremstyle{definition}
\newtheorem{definition}{Definition}

\theoremstyle{remark}
\newtheorem{remark}{Remark}
\usepackage{array}
\usepackage{epstopdf}
\usepackage{hyperref}

\usepackage{color,colortbl}
\definecolor{darkblue}{rgb}{0.0,0.0,0.65}
\definecolor{darkred}{rgb}{0.65,0.0,0.0}
\hypersetup{
  colorlinks = true,
  citecolor  = darkblue,
  linkcolor  = darkred,
  filecolor  = darkblue,
  urlcolor   = darkblue,
}

\newcommand{\Pc}{\mathsf{Par}}
\newcommand{\Qc}{\mathcal{Q}}

\newcommand{\num}{{\sf Num}} 
\newcommand{\id}{\textsf{Id}} 
\newcommand{\sat}[1]{P_{\Phi}\left(#1\right)}
\newcommand{\hsat}{\widehat{P_{\Phi}}}  
\newcommand{\mat}{M}
\newcommand{\matnor}{\boldsymbol{R}^{\Phi,d,\mathsf{rescale}}}

\newcommand{\perm}[1]{\mathbb{S}_{#1}}
\newcommand{\pii}[1]{\pi_{#1}}
\newcommand{\sii}[1]{\sigma_{#1}}

\newcommand{\tup}{U}

\newcommand{\numnew}{d_{\text{new}}}
\newcommand{\mul}{\mathsf{hist}}
\newcommand{\mult}[1]{\mathsf{hist}(#1)}
\newcommand{\vv}[1]{\mult{#1}!}

\newcommand{\Ic}{\mathcal{I}}

\DeclareMathOperator*{\tr}{Tr}   
\DeclareMathOperator*{\ex}{\mathbb{E}}
\newcommand{\norm}[1]{\left\lVert#1\right\rVert}  
\newcommand{\re}{\mathbb{R}}

\newcommand{\T}{\boldsymbol{T}^{\Phi}}  
\newcommand{\M}{\boldsymbol{M}^{\Phi}}  
\newcommand{\Sym}{\boldsymbol{R}^{\Phi,d}}
\newcommand{\Tri}{\boldsymbol{R}^{\Phi,d,\mathrm{trim}}}
\newcommand{\SPhi}{S^{\Phi}}
\newcommand{\SPhid}{S^{\Phi,d}}
\newcommand{\ff}{f^{\Phi}}

\newcommand{\OO}[1]{O\left(#1\right)} 
\newcommand{\U}[2]{U^{(#1)}_{#2}}
\newcommand{\II}[1]{I^{(#1)}}

\usepackage{thmtools} 

\declaretheorem[shaded={rulecolor=black, rulewidth=0.5pt, bgcolor=white},name=Algorithm]{algbox}

\begin{document}
	\title{A Simpler Strong Refutation of Random  $k$-XOR\thanks{This work is presented at International Conference on Randomization and Computation (RANDOM) 2020. }}
	
\author{Kwangjun Ahn}
\affil{\small{Department of Electrical Engineering and Computer Science\\ Massachusetts Institute of Technology 
\\Email: \texttt{kjahn@mit.edu} } }
	\maketitle
\begin{abstract}
Strong refutation of random CSPs  is a fundamental question in theoretical computer science that has received particular attention due to the long-standing gap  between the information-theoretic limit and the computational limit.
This gap is recently bridged by  Raghavendra, Rao and Schramm where they study sub-exponential algorithms for the regime between the two limits. 
In this work, we take a simpler approach to their algorithm and analysis.
\end{abstract}

\section{Introduction}

Refutation of random instances of constraint satisfaction problems (random CSPs) is one of the central questions in theoretical computer science with numerous applications.
Among many  predicates (types of constraints), this paper considers the XOR predicate and studies the strong refutation of the corresponding random CSP.
In fact, Allen, O’Donnell and Witmer~\cite{allen2015refute} demonstrate that one can use  strong refutation algorithms for random XOR to refute random CSPs with other predicates\footnote{For instance, it is demonstrated that one can refute random $k$-SAT by reducing it to strong refutations of random $\ell$-XOR for $\ell=1,2,\dots, k$.}. In particular, we consider:
 \begin{definition}[Random $k$-XOR]\label{def:xor} A   random $\boldsymbol{k}$-XOR with probability $p$ (or equivalently,  at density  $pn^{k-1}$) refers to a  set $\Phi =\{C_S\}$ of $k$-XOR constraints over $n$ variables $x\in\{\pm 1\}^n$ obtained as per the following procedure: 
 \begin{enumerate}
     \item   First sample each of the $n^k$ possible $k$-tuples with probability  $p$ independently.
     \item  For each sampled $S =(s_1,s_2,\dots, s_k)\in [n]^k$, include a $k$-XOR constraint $C_S$ : $\prod_{i=1}^k x_{s_i} =\eta_S$, where $\eta_S$ is i.i.d. Rademacher random variable.
 \end{enumerate}   
For an assignment $x\in\{\pm 1\}^n$, let $\sat{x}$ be the fraction of constraints satisfied by $x$.
\end{definition} 
\begin{remark}
One can alternatively consider a model where we sample subsets of size $k$ instead of $k$-tuples (there will be $\binom{n}{k}$ possible subsets in total).
However, as noted in \cite{allen2015refute}, the precise details of the random model are not relevant to the results to follow.
For simplicity, we follow the prior works~\cite{allen2015refute,raghavendra2017strongly} and  consider the above $k$-tuple model throughout the paper.
\end{remark}
Under this random $k$-XOR model,  we study the strong refutation problem.
To motivate the problem, it is a consequence of standard concentration inequalities that when the density is of $\omega(1)$ (i.e., $pn^{k-1}=\omega(1)$), with high probability, no assignment can satisfy more than  a $1/2+o(1)$ fraction of the constraints.
Hence, it is a natural algorithmic question to  ask whether one can certify such a fact.
More specifically, we consider:
\begin{definition}[Strong refutation] \label{def:strref}
	For a quantity $\alpha=\omega(1)$,
	an algorithm which takes a $k$-XOR instance and outputs a quantity $\hsat$  is said to strongly refute random $k$-XOR at density $\alpha$ if it satisfies:
	\begin{enumerate}
	    \item For any $k$-XOR instance $\Phi$, $\sat{x} \leq \hsat$ for all assignments $x\in\{ \pm 1 \}^n$.
	    \item For a random $k$-XOR instance $\Phi$ with density $\alpha$, $\hsat = 1/2+o(1)$ with high probability.
	\end{enumerate} 
\end{definition}

However, the  question of developing strong refutation algorithms for the density $\omega(1)$ turns out to be rather intractable.
More specifically, the best known guarantees  are obtained from spectral methods~\cite{allen2015refute,barak2016noisy} which require the density to be $ \widetilde{\Omega}(n^{k/2-1})$.
This computational limit of $\widetilde{\Omega}(n^{k/2-1})$ (also known as \emph{spectral threshold}) is significantly larger than the information-theoretic threshold of $\omega(1)$, and this gap has been conjectured to be fundamental.

 Recently, to bridge the gap, Raghavendra, Rao and Schramm investigate sub-exponential refutation algorithms below the spectral threshold~\cite{raghavendra2017strongly}.
 Their results constitute a smooth trade-off between the density and the time complexity required for certifying unsatisfiability.
 More specifically, their algorithm parametrized by $d$ achieves the following performance: 
 For all $\delta\in [0,1)$, their algorithm with $d=n^{\delta}$ finds a certificate at density $\widetilde{\Omega}\left(n^{(k/2-1)(1-\delta)}\right)$ in time $\exp(\widetilde{O}(n^{\delta}))$. 
 At $\delta\approx 0$, their result recovers the polynomial-time strong refutation result at the spectral threshold, while at $\delta\approx1$, their result recovers an exponential-time strong refutation at the information-theoretic threshold.

  This beautiful result, however, relies on an intricate analysis spanning over 20 pages as well as technical complications in algorithm steps, raising a question of whether one can simplify the analysis as well as the algorithm.  
  This work addresses this question as follows:
 \begin{enumerate}
     \item  This work simplifies the key technical component of the analysis in~\cite{raghavendra2017strongly} (Section~\ref{spec}). 
     More specifically, the spectral norm analysis~\cite[Theorem 4.4]{raghavendra2017strongly} is significantly simplified in this work relying  on  more elementary combinatorial arguments.
    \item In addition, for   even $k$, this work also simplifies  their refutation algorithm by modifying the technical preprocessing step (Section~\ref{algdes}).
    At a high level, the previous work requires $\OO{d}$ spectral norm computations of the matrix of size $n^{\OO{d}\times \OO{d}}$, whereas the approach in this paper only requires a \emph{single} computation.
 \end{enumerate}
As a byproduct of our simpler approach, the theoretical guarantee in this paper comes with less technical conditions and enjoys better refutation performances.

\section{Preliminary: spectral strong refutation algorithms}
\label{sec:review}
To set the stage for our main result, we first briefly review the spectral refutation algorithms in the prior works~\cite{coja2007strong,allen2015refute,barak2016noisy} that achieve the spectral threshold.
For illustrative purpose, we focus throughout on the case when $k$ is even; indeed, the odd $k$ case follows similarly modulo some extra ``tricks'' to reduce it to the even case (see e.g. \cite[Appendix A.2]{allen2015refute} for details).

We first represent the strong refutation problem as the problem of certifying an upper bound on a polynomial.
\begin{definition}[Constraints tensor] \label{def:tensor}
Given a set of constraints  $\Phi$ consisting of $m$ constraints  $C_{S_1},\dots, C_{S_m}$, the \emph{constraints tensor} of $\Phi$ is a $n^k$ tensor $\T$  defined as $\T_{S}=\eta_{S_a}$ if $S=S_a$ for some $a=1,\dots, m$ and $\T_S =0$ otherwise.    
\end{definition}
\begin{definition}[Constraints polynomial] \label{def:poly}
Given a set of constraints  $\Phi$, the \emph{constraints polynomial} of $\Phi$ is a $k$-degree homogeneous polynomial $\ff$  defined as $\ff(x):=\langle \T, x^{\otimes k}\rangle$.    
\end{definition}
Having the above definitions, it is straightforward to verify the following identity:
\begin{align} 
      &2m \cdot \left(\sat{x}-\frac{1}{2}\right) =  \sum_{\ell=1}^m \eta_{(i_\ell,j_\ell)} x_{i_\ell}x_{j_\ell} =    \langle \T, x^{\otimes k}\rangle = \ff(x) \nonumber\\
       \label{iden}    &\Longleftrightarrow \sat{x} = \frac{1}{2} + \frac{1}{2m}\cdot  
   \ff(x)\,.
\end{align}  
Having established \eqref{iden}, the strong refutation problem turns into the problem of certifying a good upper bound on the constraints polynomial:
\begin{align}  
 \max_{x\in \{\pm 1\}^n} \sat{x}  \leq \frac{1}{2}+o(1) \Longleftrightarrow     \max_{x\in \{\pm 1\}^n} \ff(x) = o(m)\,. \label{cert:k2}
\end{align} 
Now, the key idea of the spectral refutation is to certify an upper bound on the constraints polynomial by first computing a \emph{matrix representation} of the polynomial and then computing the spectral norm of the matrix\footnote{We remark that many spectral methods in the literature can be understood as following this principle of computing the spectral norm of a matrix representation. We refer readers to \cite[Section~9]{ahn2020graph} for more examples of matrix representations arising in statistical problems.}. We first formally define matrix representations:
\begin{definition}[Matrix representation] \label{def:mat}
We say an $n^{k/2}\times n^{k/2}$ matrix $M$ is a matrix representation of a degree-$k$ homogeneous polynomial $f$ if we have $f(x) = (x^{\otimes k/2})^\top M x^{\otimes k/2}$. Here and below, we use $x^{\otimes k/2}$ to denote its vector flattening\footnote{More formally, we regard $x^{\otimes k/2}$ as a vector of dimension $nk/2$ rather than  as a $n^{k/2}$ tensor.}. 
\end{definition}
If we have a matrix representation $M$ of the constraints polynomial $\ff$, one can certify an upper bound  by computing the spectral norm of the matrix representation:
\begin{align}\label{matrep}
    \max_{x\in \{\pm 1\}^n} \ff(x) = \max_{x\in \{\pm 1\}^n}  (x^{\otimes k/2})^\top M x^{\otimes k/2} \leq n^{k/2} \norm{M}\,,
\end{align}
where the inequality follows since $\norm{x^{\otimes k/2}} = \sqrt{n^{k/2}}$.

Having \eqref{matrep}, it is now crucial to find a  matrix representation that results in a small spectral norm.
It turns out that to achieve the spectral threshold, a simple matrix representation suffices.
Let us denote by $\M$  the natural $n^{k/2}\times n^{k/2}$ flattening  of the constraints tensor $\T$. 
Certainly $\M$ is a matrix representation, and hence, its symmetrization is also a matrix representation:
\begin{definition}[Symmetric matrix representation] \label{def:symrep} $\SPhi := \frac{1}{2}[\M +(\M)^\top]$.
\end{definition} 
Indeed, it follows from a standard result in random matrix theory that the symmetric matrix representation $\SPhi$ constructed from random $k$-XOR has the spectral norm $o(m)$ with high probability  as soon as the density is above the spectral threshold, i.e.,  $\alpha= \widetilde{\Omega}(n^{k/2-1})$ (see e.g. \cite[Appendix A.1]{allen2015refute} for precise details).

Thus far, we present the spectral refutation algorithms in the prior arts that achieve the spectral threshold.
Now, we move on to the result due to  Raghavendra, Rao and Schramm~\cite{raghavendra2017strongly}.
It turns out that for strong refutation  below the spectral threshold, one needs to rely on a higher-order symmetry. This will be the subject of the next section.

\section{Higher-order symmetry for refutation below spectral threshold}\label{sec:high}

In this section, we discuss the approach based on a higher-order symmetry due to Raghavendra, Rao and Schramm~\cite{raghavendra2017strongly}.
We remark that a similar technique was independently developed by Bhattiprolu,  Guruswami and Lee~\cite{bhattiprolu2017sum} under the context of finding an upper bound certificate of the  tensor injective norm.
\subsection{Higher-order type-symmetric matrix representation}
\label{subsec:sym}

To illustrate the main idea, we first define the types of the entries:
\begin{definition}[Histogram tuples]
  Let $\mult{I}$ be the $n$-tuple $(\alpha_1,\dots, \alpha_n)$ such that $\alpha_i$ is the number of times $i$ appears in $I$, i.e., the histogram of the tuple $I$. Let $\mult{I}! := \prod_{i=1}^n (\alpha_i)!$, where $0!=1$ by convention. 
\end{definition}
\begin{definition}[Types of entries]
Given a matrix representation $M$ of a constraints polynomial $\ff$,  we say two entries $M_{I,J}$ and $M_{I',J'}$ have the same \emph{type} if  $\mult{I}=\mult{I'} $ and $\mult{J}=\mult{J'}$, i.e., for all $i\in[n]$, the number of  $i$'s appearing in  $I$ (resp. $J$) is equal to that in $I'$ (resp. $J'$) .
\end{definition}
With this definition, one can easily notice that the entries of the same type corresponds to the coefficient of the same monomial in $\ff$. 
Now the key idea of \cite{raghavendra2017strongly} is to consider a matrix representation  which distributes  the coefficient of a monomial in $\ff$  \emph{equally} across the corresponding type of entries.
It turns out that such a matrix representation has small spectral norm, resulting in a better refutation certificate.
 
\begin{definition}[Type-symmetric matrix representation]
We say a matrix representation  is   type-symmetric if the entries of the same type have the same value. 
\end{definition}
To maximize the gain from a type-symmetric matrix representation, \cite{raghavendra2017strongly}  indeed considers a higher order matrix representation, which amounts to working with  $(\ff)^d$ instead of $\ff$ for some $d>1$ at the cost of increased computational complexity. 
Given  a type-symmetric matrix representation $\Sym$  of $(\ff)^d$  (we defer the formal definition to Definition~\ref{def:matsym}), we have $\ff(x)^d = (x^{\otimes kd/2})^\top \Sym x^{\otimes kd/2}$  from which one can conclude
\begin{align}\label{highorder}
    \max_{x\in \{\pm 1\}^n} \ff(x) = n^{k/2}\cdot  \max_{x\in \frac{1}{\sqrt{n}}\cdot \{\pm 1\}^n}\left[(x^{\otimes kd/2})^\top \Sym x^{\otimes kd/2} \right]^{1/d}     \leq  n^{k/2} \cdot \norm{ \Sym}^{1/d}\,.
\end{align}
  However, as mentioned in \cite[Section 4]{raghavendra2017strongly}, it turns out that the inequality in \eqref{highorder} is not tight enough for the desired result.
  To overcome this issue,   \cite{raghavendra2017strongly}  suggested the technique of removing high multiplicity rows/columns. This will be the subject of the next subsection.

\subsection{Overcoming challenge with trimming rows/columns} \label{subsec:trim}
 Before getting into the technique in \cite{raghavendra2017strongly}, let us first discuss why the inequality in \eqref{highorder} is not tight. 
 The main reason for the looseness is the fact that the left hand side of the inequality is the maximum over the specific unit vectors of the form $\frac{1}{\sqrt{n}}\cdot \{\pm 1\}^n$, while the spectral norm certificate finds the maximum over all unit vectors.
 In particular, if the maximum  of the spectral norm is achieved by a sparse vector, this certificate would no longer provide a good upper bound.

To cope with this issue, \cite{raghavendra2017strongly}  employs the \emph{trimming} step, in which they remove rows and columns of $\Sym$ corresponding to index tuples with high multiplicities, i.e., $I$'s such that coordinate values of $\mult{I}$ are large.
This technical step indeed results in a better spectral norm bound as we shall see in Section~\ref{sec:trim}.

\subsection{Technical challenge of the approach in Raghavendra-Rao-Schramm}
\label{sec:chall}

However, it turns out that analyzing this higher-order method with the trimming step is rather technical:
\begin{enumerate}
     \item  Note that the construction of symmetric matrix representation results in a rather complicated dependency structure across entries, making it hard to analyze its spectral norm.
     Indeed, the spectral norm analysis~\cite[Theorem 4.4]{raghavendra2017strongly} constitutes the main technical component of the analysis in \cite{raghavendra2017strongly}.
     
     \item Moreover, it turns out that justifying the validity of the trimming step also requires some technical modification of the algorithm together with an additional careful analysis.
     At a high level, these complications arise due to the fact that the trimmed matrix is no longer a matrix representation of the constraints polynomial. 
     In particular, their approach requires  computations of $\OO{d}$ spectral norms of matrices of size $n^{\OO{d}\times \OO{d}}$.
 \end{enumerate}
We will address the above challenges in order in the subsequent sections.

\section{A simpler spectral norm analysis} 

   \label{spec}
In this section, we provide a simpler spectral norm analysis of the symmetric matrix representation.    
As we mentioned in the previous section, the symmetric matrix representation has an intricate dependency structure between entries and hence the standard tools such as matrix Chernoff bound~\cite{tropp2012user} does not apply.
Hence, we need to rely on more direct analysis based on the trace power method:
	\begin{proposition}[Trace power method]
	\label{tracepower}
		Let $n,\ell\in \mathbb{N}$, let $c\in \re$, and let $M$ be a symmetric $n\times n$ random matrix. Then,
		\begin{align*}
		\ex \tr(M^{2\ell}) \leq \beta ~\Longrightarrow~ \Pr\left(\norm{M} \leq c\cdot \beta^{1/2\ell}\right)\geq 1-c^{-2\ell}.
		\end{align*}
	\end{proposition}
\begin{proof}
	The proof follows from the fact that $\norm{M}^{2\ell} = \norm{M^{2\ell}} \leq \tr(M^{2\ell})$ together with Markov's inequality.
\end{proof} 
\noindent Hence, to come up with a probabilistic upper bound on the spectral norm, one need to bound the trace power term.
However, in contrast to well-known settings in random matrix theory, our matrix of interest $\mat$ has $kd/2$-tuples for its row/column indices, which renders computing the trace power term more complicated.
In particular, for an integer $\ell$, the trace power term can be represented as
\begin{align}\label{trace}
\sum_{\substack{\II{1},\ldots, \II{2\ell} \in [n]^{kd/2}}} \ex\left[ \prod_{j=1}^{2\ell} \mat_{\II{j}, \II{j+1}}\right]\,, 
\end{align}
	where indices are read modulo-$2\ell$, i.e.,  $I_{2\ell+1}$ denotes $I_1$.
As a warm-up, we first analyze the symmetric matrix representation, i.e. $M=\Sym$. 

\subsection{Warm-up: analysis for higher-order type-symmetric matrix}
\label{subsec:warm-up}
In this section we apply the trace power method to $M=\Sym$ as a warm-up. Let us first formally define $\Sym$. To that end, we first recall the symmetric matrix representation $\SPhi$.
By its definition (Definition~\ref{def:symrep}), $\SPhi$ is a  $n^{k/2}\times n^{k/2}$ symmetric random matrix with independent mean-zero entries taking values in $[-1,1]$.
Now, let $\SPhid$ be the $d$-th Kronecker power   of $\SPhi$, i.e., for $k/2$-tuples $U_1,\dots, U_d $ and $V_1,\dots, V_d$, 
\begin{align}\label{def:sym}
   \SPhid_{(U_1,\dots, U_d),(V_1,\dots, V_d)} =   \SPhi_{U_1,V_1}\times \SPhi_{U_2,V_2} \times \dots  \times \SPhi_{U_d,V_d}\,.
\end{align}
Now, the symmetric matrix representation is obtained from $\SPhid$ by replacing each entry with the average of the entries of the same type as the corresponding entry.
To formally define, we begin with some notations:
\begin{definition}[Permutations]
For each positive integers $n,q$ and $I= (i_1,\ldots, i_{q})\in [n]^{q}$,  let $\perm{q}$ be the set of permutations on $[q]$. 
	For a permutation $\pi\in \perm{q}$ and a subtuple $U=(i_{j_1},\dots, i_{j_\ell})$ of $I$, let $\pi(U):= (i_{\pi(j_1)},\ldots, i_{\pi(j_\ell)})$. 
\end{definition}
Now based on these notations, we formally define $\Sym$ as follows:
\begin{definition}[Higher-order symmetric matrix representation]\label{def:matsym}
For an even integer $k$ and $d\geq 1$, $\Sym$ is an $n^{kd/2\times kd/2}$ matrix representation of $(\ff)^d$ defined as 
\begin{align}
    \Sym_{I,J} = \frac{1}{|\perm{kd/2}|^{2}} \sum_{\pi,\sigma \in \perm{kd/2} } \SPhid_{\pi(I),\sigma(J)}\,. 
\end{align}
\end{definition}
Now having the formal definition of $\Sym$, one can write the trace power term \eqref{trace} as follows (where we write each $kd/2$-tuple as $\II{\cdot }=(\U{\cdot}{1},\U{\cdot}{2},\dots, \U{\cdot}{d})$):
\begin{align}
 &\frac{1}{|\perm{kd/2}|^{4\ell}}\cdot   \sum_{\substack{\II{j} \in [n]^{kd/2}\\ j=1,\dots,2\ell }} \sum_{\substack{\pii{j},\sii{j}\in \perm{kd/2} \\ j=1,\dots, 2\ell }} \ex\left[ \prod_{j=1}^{2\ell} \SPhid_{\pii{j}(\II{j}), \sii{j}(\II{j+1})}\right] \nonumber\\
 =&\frac{1}{|\perm{kd/2}|^{4\ell}}\cdot   \sum_{\substack{\II{j} \in [n]^{kd/2}\\ j=1,\dots,2\ell }} \sum_{\substack{\pii{j},\sii{j}\in \perm{kd/2} \\ j=1,\dots, 2\ell }}  \ex\left[ \prod_{j=1}^{2\ell} \prod_{s=1}^d \SPhi_{\pii{j}(\U{j}{s}), \sii{j}(\U{j+1}{s})}\right]\,, \label{trace:2}
\end{align}
Although \eqref{trace:2} looks quite complicated, note that one can actually simplify it further.
\begin{definition}[Partition of the index set according equality] \label{def:part}
Given $\{\II{j}\} $, $\{\pii{j}\} $ and  $\{ \sii{j}\}$ ($j=1,\dots, 2\ell$), we define $\Pc(\{\II{j}\},\{\pii{j}\} , \{\sii{j}\}  )$ to be the partition of the index set $\Ic:=\{(j,s)~:~j=1,\dots,2\ell,s=1,\dots, d\}$ according to the equivalence relation
\begin{align*}
    \text{$(j,s)\sim (j',s')$ $\Leftrightarrow$ }\begin{cases} \text{$(\pii{j}(\U{j}{s}), \sii{j}(\U{j+1}{s})) = (\pii{j'}(\U{j'}{s'}), \sii{j'}(\U{j'+1}{s'}))$ or}\\
    \text{$(\pii{j}(\U{j}{s}), \sii{j}(\U{j+1}{s})) = (\sii{j'}(\U{j'+1}{s'}), \pii{j'}(\U{j'}{s'}))$.}
    \end{cases}
\end{align*}
We denote by $|\Pc(\{\II{j}\},\{\pii{j}\} , \{\sii{j}\}  )|$ the number of equivalence classes in the partition.
\end{definition}
Since $\SPhi$ is a symmetric random matrix with mean zero entries, it follows that the summand in \eqref{trace:2} corresponding to $\{\II{j}\} $, $\{\pii{j}\} $ and  $\{ \sii{j}\}$ is equal to zero if  the partition $\Pc(\{\II{j}\},\{\pii{j}\} , \{\sii{j}\}  )$ contains an equivalence class of odd size.

Hence, in order to have a nonzero summand, every equivalence class of $\Pc(\{\II{j}\},\{\pii{j}\} , \{\sii{j}\}  )$  must have even size.
\begin{definition}
Given $\{\II{j}\} $, $\{\pii{j}\} $ and  $\{ \sii{j}\}$ ($j=1,\dots, 2\ell$), we say the partition of the index set  $\Pc(\{\II{j}\},\{\pii{j}\} , \{\sii{j}\}  )$  is called an \emph{even}  partition if all  equivalence classes have even size.
\end{definition}
When $\Pc(\{\II{j}\},\{\pii{j}\} , \{\sii{j}\}  )$ is even, since each entry of $\SPhi$ is in $[-1,1]$, one can easily upper bound the summand of \eqref{trace:2} explicitly in terms of the number of equivalence classes:
\begin{align}\label{summand}
\ex\left[ \prod_{j=1}^{2\ell} \prod_{s=1}^d \SPhi_{\pii{j}(\U{j}{s}), \sii{j}(\U{j+1}{s})}\right]  \leq  p^{|\Pc(\{\II{j}\},\{\pii{j}\} , \{\sii{j}\}  )|}\,.
\end{align}
Using the upper bound~\eqref{summand}, and grouping the trace power term so that each group contains the summand corresponding to the same partition, we obtain
\begin{align}\label{trace:3}
 \eqref{trace:2}\leq   \frac{1}{\left|\perm{kd/2}\right|^{4\ell}} \sum_{\Qc\text{: even} }\left[ p^{\left|\Qc \right|} \cdot\num(\Qc)  \right]\,,
\end{align}
where $\num(\Qc):= \left|\left\{ (\{\II{j}\},\{\pii{j}\} , \{\sii{j}\}) ~:~\Pc(\{\II{j}\},\{\pii{j}\} , \{\sii{j}\}  )=\Qc \right\} \right|$.
Therefore, to upper bound the trace power term, one needs to upper-estimate  $\num(\Qc)$ for each $\Qc$.
Although the counting $\num(\Qc)$ looks complicated,   the symmetry  saves the day.
\begin{definition}$\num(\Qc~|~ \{\pii{j}\}):= \left|\left\{ (\{\II{j}\}, \{\sii{j}\}) ~:~\Pc(\{\II{j}\},\{\pii{j}\} , \{\sii{j}\}  )=\Qc \right\} \right|$.
\end{definition}
First, one can easily verify the following based on a simple symmetry argument (here  $\id$  denotes the identity permutation in $\perm{kd/2}$):
\begin{claim}\label{cl1}
	$\num\left(\Qc~|~\{\id \}\right)=\num\left(\Qc~|~\{\pii{j} \}\right) $ for any $\{\pii{j}\}$. 
\end{claim}
\begin{proof}  See Section~\ref{pf:lem1}. 
\end{proof}
Due to Claim~\ref{cl1}, it follows that:
\begin{align}
    \num\left(\Qc\right) = \left|\mathbb{S}_{kd/2}\right|^{2\ell}\cdot \num\left(\Qc~|~\{\id \}\right)\,.
\end{align}
Hence, with this argument, we reduce the problem of counting triples $(\{\II{j}\}, \{\pii{j}\}, \{\sii{j}\} )$ into the problem of counting pairs $(\{\II{j}\} , \{\sii{j}\} )$.
Now let us further reduce the problem.
To that end, we first define:
\begin{definition}
 We say a collection of index tuples $\{\II{j}\}$ is $\Qc$-\emph{valid} if there exist $\{\sii{j}\}$ such that $\Pc(\{\II{j}\}, \{\id\}, \{\sii{j}\} )$ is equal to $\Qc$.  
\end{definition}
\begin{claim}\label{cl2}
For any $\Qc$-valid $\{\II{j}\}$, there are at most  $(kd/2)^{k|\Qc|/2} \cdot \prod_{j=1}^{2\ell}\vv{\II{j}}$ different 
$\{\sii{j}\}$'s such that  $\Pc\left(\{\II{j}\}, \{\id\}, \{\sii{j}\}\right)= \Qc$.
\end{claim}
\begin{proof} 
The proof is based on an elementary counting argument. See Section~\ref{pf:lem2}.
\end{proof}
Due to Claim~\ref{cl2}, now we have:
\begin{align}
    \num\left(\Qc\right) \leq  \left|\mathbb{S}_{kd/2}\right|^{2\ell}\cdot  (kd/2)^{k|\Qc|/2} \cdot \sum_{\{\II{j}\}~:~\Qc\text{-valid}} \left[\prod_{j=1}^{2\ell}\vv{\II{j}}\right] \,.
\end{align}
Putting this back to \eqref{trace:3}, we obtain the following result:
\begin{theorem} \label{thm:basic}For even $k$ and $d\geq 1$, let $\Sym$ be the $n^{kd/2\times kd/2}$ higher-order symmetric matrix representation (Definition~\ref{def:matsym}) of random $k$-XOR. Then, the following upper bound on the trace power term holds:
\begin{align*}
    \ex \tr((\Sym)^{2\ell}) \leq \frac{1}{\left|\perm{kd/2}\right|^{2\ell}} \sum_{\Qc\text{: even} }\left[  \left(p(kd/2)^{k/2}\right)^{|\Qc|} \cdot \sum_{\substack{\{\II{j}\}~:\\\Qc\text{-valid}}} \left[\prod_{j=1}^{2\ell}\vv{\II{j}}\right]   \right]\,.
\end{align*}
\end{theorem} 
Having established  Theorem~\ref{thm:basic}, one can slightly modify the proof to handle the trimmed matrix from Section~\ref{subsec:trim}.  This will be the focus of the next subsection.
\subsection{A simpler spectral norm analysis of  the trimmed matrix}\label{sec:trim}
Having established Theorem~\ref{thm:basic}, which explicitly characterizes  the upper bound on the trace power term in terms of $\vv{\II{j}}$'s, one can now quantitatively understand the trimming technique due to Raghavendra, Rao and Schramm~\cite{raghavendra2017strongly}.
In particular, we will shortly demonstrate that our Theorem~\ref{thm:basic} recovers the main technical result~\cite[Theorem 4.4]{raghavendra2017strongly}.
This is  remarkable as our proof is much simpler than the original proof in \cite{raghavendra2017strongly}.

The problem with the upper bound in Theorem~\ref{thm:basic} is that the value $\vv{\II{j}}$ could be in general large.
For instance, if $\II{j}$ is the $kd/2$-tuple consisting only of index $1$, then $\vv{\II{j}} = (kd/2)!$, which turns out to be too large for our desired result. 
Now having observed this, one can now see how the trimming preprocessing of \cite{raghavendra2017strongly} helps reduce the spectral norm: by removing rows/columns corresponding to the index tuples with high multiplicities, one can significantly reduce the upper bound.
More formally, following \cite{raghavendra2017strongly}, if we remove the rows/columns corresponding to the index tuples $I$'s such that $\mult{I}$ has a coordinate value larger than $10 \log n$, we have the following:
\begin{corollary}\label{thm:trim}
For even $k$ and $d\geq 1$, let $\Tri$ be the $n^{kd/2\times kd/2}$ matrix obtained from the  $\Sym$ (Definition~\ref{def:matsym}) by removing all rows/columns $I$'s such that $\mult{I}$ has a coordinate value larger than $10 \log n$. Assume that  $d^{k/2-1}n^{k/2}p>1$. Then, the following spectral norm bound holds with   probability at least $1-n^{-2}$: 
\begin{align*}
    \norm{\Tri}^{1/d} \leq c\cdot \frac{e^{3k/4}\cdot 10^{5k/2}}{(k/2)^{k/4}} \cdot \frac{n^{k/4}p^{1/2}}{d^{ (k-2)/4}} \cdot \log^{5k/2+1} n \,.
\end{align*}
for some absolute constant $c>0$.
\end{corollary}
\begin{remark} \label{rmk:odd}
Although we focus on the  even $k$ case throughout the proof for simplicity, we note that a similar argument applies to the case of odd $k$ following the ``tricks''~\cite[Section 4.2]{raghavendra2017strongly} based on Cauchy Schwartz inequality.
Consequently, our proof technique   provides a simpler proof of the main technical statement for the odd $k$ case~\cite[Theorem 4.13]{raghavendra2017strongly}
\end{remark}
\begin{proof}[Proof of Corollary~\ref{thm:trim}]
    From Theorem~\ref{thm:basic}, we have the following upper  bound on the trace power term: 
    \begin{align*} 
    \ex \tr((\Tri)^{2\ell}) \leq \frac{1}{\left|\perm{kd/2}\right|^{2\ell}} \sum_{\Qc\text{: even} }\left[  \left(p(kd/2)^{k/2}\right)^{|\Qc|} \cdot \sum_{\substack{\{\II{j}\}~:\\\Qc\text{-valid}}} \left[\prod_{j=1}^{2\ell}\vv{\II{j}}\right]   \right]\,.
\end{align*}
On the other hand, due to the trimming procedure, each coordinate value of the tuple $\mult{\II{j}}$ is upper bounded by $10\log n$, from which we have the following upper bound on the $\vv{\II{j}}$:
\begin{align*}
    \vv{\II{j}} \leq \left((10 \log n)!\right)^{kd/2} \leq (10\log n)^{5kd\log n} \leq n^{5kd\log(10\log n)}\,.
\end{align*}
The trimming step gives us  an uniform upper bound on $\vv{\II{j}}$, and hence, it suffices to upper bound the number of $\Qc$-valid $\{\II{j}\}$'s:
 \begin{claim} \label{cl3}
 	For any even partition $\Qc$, there are at most $n^{k(|\Qc|+d)/2}$ $\Qc$-valid $\{\II{j}\}$'s.
 \end{claim}
\begin{proof}
The proof is  elementary. See Section~\ref{pf:lem3}.
\end{proof}
Due to Claim~\ref{cl3}, the upper bound on the trace power term becomes:
  \begin{align} 
    &\ex \tr((\Tri)^{2\ell}) \leq \frac{n^{5kd\log(10\log n)}}{\left|\perm{kd/2}\right|^{2\ell}} \sum_{\Qc\text{: even} }\left[  \left(p(kd/2)^{k/2}\right)^{|\Qc|} \cdot n^{k(|\Qc|+d)/2} \right] \nonumber \\
    \label{trace:trim2}&= \frac{n^{5kd\log(10\log n)+kd/2}}{\left|\perm{kd/2}\right|^{2\ell}} \sum_{M=1 }^{d\ell}\left[ N_M \cdot  \left(p(nkd/2)^{k/2}\right)^{M} \right]\,,
\end{align}
where $N_M$ is the number of even partitions of size $M$ and  we have $M\leq d\ell$ in the range of summation since an even partition has size at most $d\ell$.
Thus, the last ingredient is to bound the number of even partitions:
\begin{claim}\label{cl4}
$N_M\leq \binom{2d\ell}{M} \cdot M^{2d\ell-M}$ for all $1\leq M\leq d\ell$. 
\end{claim}
\begin{proof}The first term in the upper bound accounts for the number of different ways of choosing $M$ representative indices in $\Ic$, and the second term counts the number of different ways of assigning the other indices to the $M$ representative elements.
\end{proof}
Due to Claim~\ref{cl4}, the upper bound \eqref{trace:trim2} becomes:
\begin{align} \label{trace:trim3}
    \frac{n^{5kd\log(10\log n)+kd/2}}{\left|\perm{kd/2}\right|^{2\ell}} \sum_{M=1 }^{d\ell}\left[ \binom{2d\ell}{M} \cdot M^{2d\ell-M} \cdot  \left(p(nkd/2)^{k/2}\right)^{M} \right]\,.
\end{align}

Having established \eqref{trace:trim3}, the rest of the proof is straightforward calculations.
We first upper bound each term in the above summand as follows:
 (i) $\binom{2d\ell}{M} \leq 2^{2d\ell}$, (ii) $M^{2d\ell-M}\leq (d\ell)^{2d\ell-M} \leq d^{2d\ell-M} \ell^{2d\ell}$, and (iii) $ \left((kdn/2)^{k
 	/2} p \right)^{M}\leq \left((dn)^{k/2}p\right)^M \cdot (k/2)^{kd\ell/2}$.
 Then, the summand in \eqref{trace:trim3}  is upper bounded by
 \begin{align*}
     2^{2d\ell} \cdot d^{2d\ell-M} \ell^{2d\ell} \cdot  (k/2)^{kd\ell/2}\left((dn)^{k/2}p\right)^M   =(2\ell)^{2d\ell} (k/2)^{kd\ell/2}\cdot d^{2d\ell}\left(d^{k/2-1}n^{k/2}p\right)^M  \,.
 \end{align*}
 Using this upper bound, it follows that
 	\begin{align}
 	\eqref{trace:trim3}&\leq  \frac{(2\ell)^{2d\ell}  (k/2)^{kd\ell/2} n^{5kd\log(10\log n)+kd/2} d^{2d\ell}}{\left|\perm{kd/2}\right|^{2\ell}} \cdot\sum_{M=1}^{d\ell} \left(d^{k/2-1}n^{k/2}p\right)^M \nonumber\\
 &\leq \frac{(2\ell)^{2d\ell}   (k/2)^{kd\ell/2} e^{kd\ell} n^{5kd\log(10\log n)+kd/2} d^{2d\ell} }{(kd/2)^{kd\ell}} \cdot d\ell \cdot \left(d^{k/2-1}n^{k/2}p\right)^{d\ell}\,, \label{trace:trim4}
 	\end{align}
 where the inequality follows from the facts that $|\perm{kd/2}|^{2\ell}  = \left((kd/2)!\right)^{2\ell} \geq (\frac{kd/2}{e})^{kd\ell}$ ($\because n!\geq (n/e)^n$) and $d^{k/2-1}n^{k/2}p>1$.
  Reorganizing terms in \eqref{trace:trim4}, we obtain
  \begin{align*}
  (2\ell)^{2d\ell+1} e^{kd\ell} (k/2)^{-kd\ell/2}  d^{-kd\ell/2 +d\ell +1} n^{kd\ell/2+5kd\log(10\log n)+kd/2} p^{d\ell}\,.
  \end{align*} 
Invoking Proposition~\ref{tracepower} and using the fact that  $f(x)=x^{1/x}$ is bounded on $[1,\infty)$, 
$\norm{M}^{1/d}$ is upper bounded by 
\begin{align*}
c\cdot (2\ell)e^{k/2} (k/2)^{-k/4} d^{-k/4 +1/2}n^{k/4 +5k\log(10\log n)/(2\ell)+k/(4\ell)}p^{1/2}
\end{align*}
with probability at least $1-e^{-2\ell}$
 for some absolute constant $c>0$.
Choosing $\ell =\log n$, we complete the proof.
\end{proof}
Thus far, we have  addressed the first challenge in Section~\ref{sec:chall} by developing a simpler spectral norm analysis of the type-symmetric representation as well as the trimmed matrix.
Now, we move on to the second challenge: as mentioned in Section~\ref{sec:chall}, the trimmed matrix $\Tri$ is no longer a matrix representation of $(\ff)^d$, it requires additional non-trivial modifications of the algorithm steps as well as analysis.

\section{A simpler spectral refutation with re-scaling entries} \label{algdes}
In this section, we address the second challenge from Section~\ref{sec:chall} and develop a simpler spectral refutation algorithm.
Our main idea is to re-scale the rows/columns of $\Sym$.
To describe our re-scaling step, we first revisit the upper bound from Theorem~\ref{thm:basic}:
\begin{align}\label{trace:upper}
    \ex \tr((\Sym)^{2\ell}) \leq \frac{1}{\left|\perm{kd/2}\right|^{2\ell}} \sum_{\Qc\text{: even} }\left[  \left(p(kd/2)^{k/2}\right)^{|\Qc|} \cdot \sum_{\substack{\{\II{j}\}~:\\\Qc\text{-valid}}} \left[\prod_{j=1}^{2\ell}\vv{\II{j}}\right]   \right]\,.
\end{align}
As we have discussed in Section~\ref{sec:trim},  we need to cancel out the  $\vv{\II{j}}$ terms in the bound to reduce the spectral norm.
Our approach is to appropriately re-scale $\Sym$ so that one can remove the $\prod_{j=1}^{2\ell}\vv{\II{j}}$ terms in the upper bound~\eqref{trace:upper}. 
In particular, if we divide the $(I,J)$-th entry of $\Sym$ by $\sqrt{\mult{I}!\cdot \mult{J}!}$, the   $\prod_{j=1}^{2\ell}\vv{\II{j}}$ term will be exactly canceled out by the re-scaling.
More formally, we define the following vector and its corresponding diagonal matrix:
\begin{definition}[Re-scaling factors] \label{def:refac}
 Let $\mul$ be an $n^{kd/2}$-dimensional vector  whose $I$-th coordinate is defined as $\mul_I := \sqrt{\mult{I}!}$ for each $I\in [n]^{kd/2}$.
 We define $D_\mul$ to be an $n^{kd/2}\times n^{kd/2}$ diagonal matrix whose $(I,I)$-th entry is defined as $\mul_I$. 
\end{definition}
Using Definition~\ref{def:refac}, one can precisely achieve the re-scaling discussed above as follows:
\begin{definition}[Re-scaled matrix representation] \label{def:rescale}
 $\matnor:= D_\mul^{-1} \cdot \Sym \cdot D_\mul^{-1}$.
\end{definition} 
Then, following the same proof as  that of Corollary~\ref{thm:trim}, one can prove the following spectral norm bound:
\begin{corollary}\label{thm:rescale}
For even $k$ and $d\geq 1$, let $\matnor$ be the $n^{kd/2\times kd/2}$ matrix obtained from the  $\Sym$ by re-scaling the rows/columns as per \eqref{def:rescale}. 
Assume that  $d^{k/2-1}n^{k/2}p>1$. Then, the following spectral norm bound holds with   probability at least $1-n^{-2}$: 
\begin{align*}
    \norm{\matnor}^{1/d} \leq c\cdot \frac{e^{3k/4}
    }{(k/2)^{k/4}} \cdot \frac{n^{k/4}p^{1/2}}{d^{ (k-2)/4}} \cdot \log n \,.
\end{align*}
for some absolute constant $c>0$.
\end{corollary}
\begin{remark}
Note that the spectral norm bound in Corollary~\ref{thm:rescale} is better than the bound due to the trimming step (Corollary~\ref{thm:trim}). This improvement actually leads to a better strong refutation guarantee as we shall see in Theorem~\ref{thm:main}. Also see Section~\ref{comparison} for an extensive comparison with \cite{raghavendra2017strongly}.
\end{remark}
\begin{proof}
Due to the re-scaling factor, following the proof of Theorem~\ref{thm:basic}, we obtain the following bound on the trace bower term without the $\prod_{j=1}^{2\ell}\vv{\II{j}}$ term:
\begin{align*}
    \ex \tr((\matnor)^{2\ell}) \leq \frac{1}{\left|\perm{kd/2}\right|^{2\ell}} \sum_{\Qc\text{: even} }\left[  \left(p(kd/2)^{k/2}\right)^{|\Qc|} \cdot \sum_{\substack{\{\II{j}\}~:\\\Qc\text{-valid}}} \left[1 \right]   \right]\,.
\end{align*}
Now due to Claims~\ref{cl3} and \ref{cl4}, one can further upper bound the trace power term by
\begin{align} \label{trace:rescale}
\frac{n^{kd/2}}{\left|\perm{kd/2}\right|^{2\ell}} \sum_{M=1 }^{d\ell}\left[ \binom{2d\ell}{M} \cdot M^{2d\ell-M} \cdot  \left(p(nkd/2)^{k/2}\right)^{M} \right]\,,
\end{align}
which is better than \eqref{trace:trim3} by a multiplicative factor of $n^{5kd\log(10\log n)}$.
Now, following the exact same calculations as in the proof of Corollary~\ref{thm:trim} and choosing $\ell=\log n$, one can easily notice that the improvement by a multiplicative factor of $n^{5kd\log(10\log n)}$ results in an improvement in the final bound by a multiplicative factor of $n^{5k\log(10\log n)/(2\ell)}=n^{5k\log(10\log n)/(2\log n)} = (10\log n)^{5k/2}$, which completes the proof.
\end{proof}
With this re-scaled matrix $\matnor$, one can also easily come up with a valid certificate for strong refutation (Definition~\ref{def:strref}): 
\begin{proposition} \label{prop1}
	For any $k$-XOR instance $\Phi$ and assignment $x\in\{\pm 1 \}^n$, we have 
	\begin{align*}
	\left|\sat{x} -\frac{1}{2}\right|\leq  \frac{1}{2m} \left[\|\matnor\|\cdot \left(\sum_{I\in [n]^{kd/2}}\vv{I}\right) \right]^{1/d}\,.
	\end{align*}
	In other words, $\frac{1}{2}+\frac{1}{2m} [\|\matnor\|\cdot (\sum_{I\in [n]^{kd/2}}\vv{I}) ]^{1/d}$ is a valid certificate for strong refutation. 
\end{proposition} 
\begin{proof} First, since $\Sym$ is a matrix representation of $(\ff)^d$, we have
\begin{align*}
    \ff(x)^d = (x^{\otimes kd/2})^\top \Sym x^{\otimes kd/2}\,.
\end{align*}
Hence, it follows that
\begin{align*}
    \ff(x)^d &= (D_{\mul}x^{\otimes kd/2})^\top\cdot  D_\mul    D_{\mul}^{-1} \cdot \Sym \cdot  D_{\mul}^{-1}  D_{\mul} \cdot x^{\otimes kd/2} \\
    &= (D_{\mul}x^{\otimes kd/2})^\top\cdot \matnor\cdot  D_{\mul}x^{\otimes kd/2} \,.
\end{align*}
Consequently, we have 
\begin{align*}
    |\ff(x)^d| \leq \norm{\matnor} \cdot \norm{D_{\mul}x^{\otimes kd/2}}^2 = \norm{\matnor} \cdot \left(\sum_{I\in [n]^{kd/2}}\vv{I}\right)\,,
\end{align*}
where the equality is due to the fact that $x^{\otimes kd/2}$ is an $n^{kd/2}$-dimensional vector with  coordinates equal to $\pm 1$.
Therefore, the proposition follows thanks to the identity \eqref{iden}, which reads $\sat{x} = \frac{1}{2} + \frac{1}{2m}\cdot  \ff(x)$.
\end{proof}
Hence, in order to guarantee that the certificate from Proposition~\ref{prop1} works, our last ingredient is to show that the term $(\sum_{I\in [n]^{kd/2}}\vv{I})$ is not too large compared to $\norm{x^{\otimes kd/2}}^2 =n^{kd/2}$. 
\begin{proposition}\label{prop2}
  For even $k$ and $d\geq 1$, 
  \begin{align*}
  \sum_{I\in [n]^{kd/2}}\vv{I} = \frac{(kd/2+n-1)!}{(n-1)!}  =  (kd/2+n-1)(kd/2+n-2)\cdots n  \,.
  \end{align*}
  In particular, if $d\leq n$, we have $  \sum_{I\in [n]^{kd/2}}\vv{I} \leq  (k/2+1)^{kd/2}n^{kd/2}$.
\end{proposition}
\begin{proof}
  We first group the terms in the summation according to the value of $\mult{I}$: 
  \begin{align}\label{hist:1}
  \sum_{I\in [n]^{kd/2}}\vv{I}=\sum_{\substack{(s_1,s_2,\dots,s_n)\in (\mathbb{Z}_{\geq0})^n:\\  \sum_i s_i = kd/2}} \sum_{\substack{I\in [n]^{kd/2}~:\\\mult{I}= (s_1,s_2,\dots,s_n)}} \prod_{i=1}^n(s_i)!\,.
  \end{align}
For each $(s_1,s_2,\dots,s_n)\in (\mathbb{Z}_{\geq0})^n$, there are $\frac{(kd/2)!}{\prod_{i=1}^n(s_i)!}$ different $I$'s such that $\mult{I} = (s_1,s_2,\dots,s_n)$.
Hence, the right hand side of \eqref{hist:1} becomes
  \begin{align*} \sum_{\substack{(s_1,s_2,\dots,s_n)\in (\mathbb{Z}_{\geq0})^n:\\  \sum_i s_i = kd/2}} (kd/2)!=  (kd/2)!\cdot \left|\left\{(s_1,s_2,\dots,s_n)\in (\mathbb{Z}_{\geq0})^n:   \sum_i s_i = kd/2\right\}\right|\,.
  \end{align*}
It is a simple enumerative combinatorics (c.f. stars and bars argument) to show that the number of feasible $(s_1,\dots,s_n)$'s is equal to $\binom{kd/2+n-1}{n-1} =\binom{kd/2+n-1}{kd/2}$. 
Therefore, the summation is equal to 
\begin{align*}
    \binom{kd/2+n-1}{kd/2}\cdot (kd/2)! &=  (kd/2+n-1)(kd/2+n-2)\cdots n\,,
\end{align*}
which completes the proof.
\end{proof}

Combining what we have obtained thus far,  one can consider the following simpler refutation algorithm based on re-scaling entries: 
\begin{algbox} \label{alg:1} A simpler strong refutation algorithm with parameter $d$ for even $k$.
\begin{enumerate}
	\item[ ] {\bf Input:} A $k$-XOR instance $\Phi$ on $n$ variables consisting of $m$ clauses $C_{S_1},\dots, C_{S_m}$ for distinct tuples $S_1,\dots ,S_m\in[n]^k$ and a parameter $d\in\mathbb{N}$.
	\item Construct a  higher-order symmetric matrix representation $\Sym$  based on the $k$-XOR instance $\Phi$ as per Definition~\ref{def:matsym}.
	\item Compute   $\matnor$  as per Definition~\ref{def:rescale}.
	\item[] {\bf Output:} $\hsat:=\frac{1}{2}+\frac{1}{2m} \norm{\matnor}^{1/d} \cdot  \left(\frac{(kd/2+n-1)!}{(n-1)!}\right)^{1/d}  $.
\end{enumerate}
\end{algbox}  
\begin{remark}
A similar idea of re-scaling rows/columns with diagonal matrices to obtain a better certificate also appeared in the MAXCUT literature; see e.g. \cite[Theorem 2.2]{delorme1993performance}. 
\end{remark}
\begin{theorem}\label{thm:main}
	Let $d\leq n$ be positive integers and $k$ be an even integer. 
	For any instance  $\Phi$ of $k$-XOR, the output $\hsat$ of Algorithm~\ref{alg:1} satisfies  $|\sat{x}-\frac{1}{2}|\leq \hsat-\frac{1}{2}$ for any $x\in \{ \pm 1 \}^n$.
	Assume further that $\Phi$ is an instance of random $k$-XOR with probability $p$ (Definition~\ref{def:xor}).
	If $p\cdot d^{k/2-1}n^{k/2}>1$,  the following bound holds with probability at least $1-O(n^{-1})$  for some absolute constant $c>0$:
\begin{align*}
  \hsat -\frac{1}{2}  \leq    c\cdot \frac{  \log n }{\sqrt{d^{k/2 -1} n^{k/2}p}} \cdot \frac{e^{3k/4}\cdot (k/2+1)^{k/2}}{(k/2)^{k/4}}\,.
\end{align*}
In particular, Algorithm~1 with parameter $d$ certifies  with high probability that $\sat{x}$ is equal to $1/2 \pm  o(1)$ for any $x\in \{ \pm 1 \}^n$ whenever $n^{k-1}p=\omega( (n/d)^{k/2-1}\log^2n)$.
\end{theorem}
\begin{proof}
	 First from Proposition~\ref{prop1}, we have
	 \begin{align} \label{ineq:1}
	 \left|\sat{x} -\frac{1}{2}\right|\leq  \frac{1}{2m} \left[\|\matnor\|\cdot \left(\sum_{I\in [n]^{kd/2}}\vv{I}\right) \right]^{1/d} = \hsat -\frac{1}{2}\,,
	 \end{align}
	 where the equality is due to Proposition~\ref{prop2}.
	 Hence the first part of the theorem is proved.
	 As for the second part,  it follows from Corollary~\ref{thm:rescale} and Proposition~\ref{prop2} that with probability at least $1-O(n^{-2})$:
	 \begin{align} \label{ineq:2}
	    \left[\|\matnor\|\cdot \left(\sum_{I\in [n]^{kd/2}}\vv{I}\right) \right]^{1/d} \leq  c\cdot  \frac{n^{3k/4}p^{1/2}}{d^{ (k-2)/4}}  \log n \cdot \frac{e^{3k/4}
    \cdot (k/2+1)^{k/2}}{(k/2)^{k/4}}
	 \end{align}
	 for some absolute constant $c>0$.
	 Next, it follows from a standard concentration inequality (e.g. Chernoff bound) that with probability at least (say) $1-n^{-10}$, $m\geq pn^k/2$.
	 Putting these bounds back to \eqref{ineq:1}, we obtain 
	 \begin{align*} 
	 \left|\sat{x} -\frac{1}{2}\right|&\leq   \frac{1}{pn^k}\cdot c\cdot  \frac{n^{3k/4}p^{1/2}}{d^{ (k-2)/4}}  \log n \cdot \frac{e^{3k/4}
    \cdot (k/2+1)^{k/2}}{(k/2)^{k/4}} \\
    &= c\cdot  \frac{\log n}{\sqrt{d^{k/2-1} n^{k/2} p} } \cdot \frac{e^{3k/4}
    \cdot (k/2+1)^{k/2}}{(k/2)^{k/4}}\,,
	 \end{align*}
	and hence, the second part of the theorem also follows. 
\end{proof}

\section{Comparison with Raghavendra-Rao-Schramm}
 \label{comparison}
We compare  Algorithm~\ref{alg:1} with the refutation algorithm of Raghavendra, Rao and Schramm~\cite{raghavendra2017strongly}.
First, the algorithm steps in this paper is simpler than that of \cite{raghavendra2017strongly}. 
As we have discussed earlier, the trimming step in the algorithm of \cite{raghavendra2017strongly}  causes some technical complications as the resulting matrix  is no longer a matrix representation of $(\ff)^d$.
Indeed,  their algorithm first constructs  matrices of size $n^{kj/2}\times n^{kj/2}$ for $j\in [\delta d, d]$ and computes the spectral norms of those matrices to design a refutation certificate; see \cite[Section 4.1.1]{raghavendra2017strongly} for details.
This is in stark contrast with Algorithm~\ref{alg:1} which only computes the spectral norm of a \emph{single} matrix $\matnor$ of size $n^{kd/2}\times n^{kd/2}$.
In addition, while their certificate requires non-trivial analysis~\cite[Section 4.1.1]{raghavendra2017strongly} to guarantee its validity, the validity of our certificate $\hsat$ readily follows as we saw in Proposition~\ref{prop1}.

As a result of the simpler approach in this paper, the theoretical guarantee in this paper comes with less technical conditions and enjoys a better refutation guarantee as well as density requirement. 
More specifically, unlike the guarantee in \cite{raghavendra2017strongly}, our main theorem does not require a technical condition like $d\log n=\OO{n}$.
Moreover, the density requirement for strong refutation reads  $n^{k-1}p=\omega( (n/d)^{k/2-1}\log^{2k}n)$ in \cite{raghavendra2017strongly}, which is worse than  that of this paper by a poly-logarithmic factor (recall that the requirement in Theorem~\ref{thm:main} reads $n^{k-1}p=\omega( (n/d)^{k/2-1}\log^{2}n)$).
 Lastly, even when the density requirement is fulfilled, their refutation guarantee reads $\frac{1}{2}+\gamma+o(1)$ for some constant $\gamma>0$ that depends on a hyperparameter in the trimming step. 
 On the other hand, this constant $\gamma$ does not appear in the refutation guarantee of this paper.
 
 \section{Conclusion}
 
 In this paper, we establish a simpler approach to strong refutation of random $k$-XOR below the spectral threshold.
 Our simplification is two-fold.
 First, we  provide a simpler spectral norm analysis of the certificate matrix of the previous work~\cite{raghavendra2017strongly} (Section~\ref{spec}). Second,   we develop a simple strong refutation algorithm  for the even $k$ case  (Section~\ref{algdes}).
 Thanks to our simpler approach, our main result (Theorem~\ref{thm:main}) enjoys a better theoretical guarantee under less   assumptions.
It is important to note that a recent work by Wein, El Alaoui and Moore also establishes a simpler strong refutation algorithm for random even $k$-XOR~\cite[Theorem F.1]{wein2019kikuchi} with a different approach.
Given the successful simplifications for the even $k$ case, it would be  interesting  to see if one can come up with a simpler strong refutation algorithm  for the odd $k$ case.

\section*{Acknowledgements}

The author thanks Vijay Bhattiprolu for suggesting the idea of re-scaling entries of the matrix representation and anonymous reviewers for valuable comments. 
The author acknowledges  the Kwanjeong Educational Foundation scholarship and also partial support as a graduate research assistant from the NSF Grant (CAREER: 1846088). 
 
\bibliographystyle{plain}
\bibliography{ref}
\appendix

\section{Deferred proofs of claims}
\subsection{Proof of Claim~\ref{cl1}} \label{pf:lem1}
Recall that Claim~\ref{cl1} reads		$\num\left(\Qc~|~\{\id \}\right)=\num\left(\Qc~|~\{\pii{j} \}\right) $ for any $\{\pii{j}\}$.  
Let us arbitrarily fix a collection of permutations  $\{\pii{j}\}$.
The main observation is that for any $\{\II{j}\}$ and $\{\sii{j}\}$, we have
$\Pc(\{\II{j}\},\{\pii{j}\} , \{\sii{j}\}) = \Pc(\{\pii{j}(\II{j})\},\{\id\} , \{\sii{j}\circ \pii{j}^{-1}\})$.
This is a straightforward consequence of Definition~\ref{def:part}.
Hence, there is an one-to-one correspondence between the collection of pairs $(\{\II{j}\}, \{\sii{j}\})$ such that $\Pc(\{\II{j}\},\{\pii{j}\} , \{\sii{j}\}  )=\Qc$ and the collection such that $\Pc(\{\II{j}\},\{\id\} , \{\sii{j}\}  )=\Qc$. 
This concluded the proof.

\subsection{Proof of Claim~\ref{cl2}} \label{pf:lem2}
We first restate Claim~\ref{cl2}: for any $\Qc$-valid $\{\II{j}\}$, there are at most  $(kd/2)^{k|\Qc|/2} \cdot \prod_{j=1}^{2\ell}\vv{\II{j}}$ different 
$\{\sii{j}\}$'s such that  $\Pc\left(\{\II{j}\}, \{\id\}, \{\sii{j}\}\right)= \Qc$.

We bound the number of feasible $\{\sii{j}\}$'s  as we go through the index set $\Ic=\{(j,s)~:~j=1,\dots,2\ell,s=1,\dots, d\}$  in the lexicographical order, i.e., $(1,1)$, $(1,2)$, \dots, $(1,d)$,  $(2,1)$, \dots and so on.
As we read the indices in such an order,  we call an index $(j,s)$ \emph{new} if  $(\U{j}{s}, \sii{j}(\U{j+1}{s}))$ is not equivalent to the previously appeared indices.
Consider the indices $(j,1),(j,2),\dots, (j,d)$ for a fixed  $j\in [2\ell]$. 
We consider two different scenarios:
\begin{enumerate}
    \item First, suppose that all indices $(j,1),(j,2),\dots, (j,d)$ are old.
Then it should be the case that for each $m=1,\dots, kd/2$, $\sii{j}(m)$ is chosen so that the $\sii{j}(m)$-th coordinate of $\II{j+1}$ respects the previous appeared equivalent index.
Having observed this,  it readily follows that there are $\mult{\II{j+1}}!$ different choices for $\sii{j}(1),\dots \sii{j}(kd/2)$ considering the permutation. 
\item Now, suppose that there are $\numnew^{(j)}$ new indices among  $(j,1),(j,2),\dots, (j,d)$.
For simplicity, assume that $\tup_{j,1},\dots, \tup_{j,\numnew^{(j)}}$ are new.
Choosing the values $\sii{j}(1),\sii{j}(2),\dots, \sii{j}(k\numnew^{(j)}/2)$ arbitrarily, there are at most 
\begin{align*}
    (kd/2)(kd/2-1)\cdots (kd/2-k\numnew^{(j)}/2+1) \leq (kd/2)^{k\numnew^{(j)}/2}
\end{align*}
different choices for $\sii{j}(1),\sii{j}(2),\dots, \sii{j}(k\numnew/2)$.
A similar counting to previous case yields that for the remaining values there are at most $\mult{\II{j+1}}!$ different choices.
\end{enumerate}
Taking a product over all $j$'s, we complete the proof since $\sum_{m=1}^{2\ell}\numnew^{(m)} =|\Qc|$. .\qed

\subsection{Proof of Claim~\ref{cl3}} \label{pf:lem3}
Let $\Qc$ be an even partition.
We count the number of possible $\Qc$-valid $\{\II{j}\}$'s.
First, let us choose $\II{1}$ arbitrarily.
Note that there are $n^{kd/2}$ different ways of choosing $I_1$.
Now, consider $I_2,\dots, I_{2\ell}$.
Similar to the proof of Claim~\ref{cl2}, we will bound the number of feasible choices s we go through the index set $\Ic=\{(j,s)~:~j=1,\dots,2\ell,s=1,\dots, d\}$  in the lexicographical order.
Again, we call an index $(j,s)$ \emph{new} if  $(\pii{j}(\U{j}{s}), \sii{j}(\U{j+1}{s}))$ is not equivalent to the previously appeared indices.

Note that we only need to consider new indices because the tuples of old indices are fully determined by their previous appearance.
We begin with the tuples $ (1,1), (1,2),\dots, (1,d)$.
Whenever we encounter a new tuple, say $(\U{1}{s}, \sii{1}(\U{2}{s}))$ , we only need to specify $\sii{1}(\U{2}{s})$ since  $\II{1}$ is already fully specified.
Hence, there are at most $n^{k\numnew^{(1)}/2}$ different ways of choosing $\II{2}$, where $\numnew^{(1)}$ is the number of new indices among $ (1,1), (1,2),\dots, (1,d)$.
By similar arguments, inductively for $j=2,3,\dots, 2\ell$, there are at most $n^{k\numnew^{(j)}/2}$ different ways of choosing $\II{j}$.
Taken collectively, we obtain the result since $\sum_{j=1}^{2\ell}\numnew^{(j)} =|\Qc|$. 

\end{document}